\documentclass[a4paper,11pt]{article}
\usepackage{amsthm,amsfonts,amssymb,euscript}
\usepackage{latexsym, multicol, fancybox}
\usepackage{graphicx}
\usepackage{color}
\usepackage{amsmath, amsthm, amssymb, bm}
\usepackage{epstopdf}
\usepackage{caption}
\usepackage{psfrag}
\usepackage{slashed}
\usepackage{verbatim}

	\def\grad{\mbox{grad}}												

\usepackage{mathrsfs}
 \usepackage{xcolor}
 \usepackage[citebordercolor={green}]{hyperref}
 \usepackage{tikz}
 \usepackage{bbm}
 \usepackage{dsfont}
 \usepackage{bbold}
  \usepackage{cite}
  \usepackage{hyperref}
\usepackage{cleveref}
\usepackage[normalem]{ulem}
\usepackage{cancel}

\numberwithin{equation}{section}

\newtheorem{theorem}{Theorem}[section]
\newtheorem{lemma}[theorem]{Lemma}
\newtheorem{proposition}[theorem]{Proposition}

\newtheorem{definition}[theorem]{Definition}
\newtheorem{remark}[theorem]{Remark}

\newcommand{\bea}{\begin{eqnarray}}
\newcommand{\eea}{\end{eqnarray}}
\newcommand{\lab}{\lababel}

\DeclareFontFamily{U}{mathx}{\hyphenchar\font45}
\DeclareFontShape{U}{mathx}{m}{n}{
      <5> <6> <7> <8> <9> <10>
      <10.95> <12> <14.4> <17.28> <20.74> <24.88>
      mathx10
      }{}
\DeclareSymbolFont{mathx}{U}{mathx}{m}{n}
\DeclareFontSubstitution{U}{mathx}{m}{n}
\DeclareMathAccent{\widecheck}{0}{mathx}{"71}

\def\beaa{\begin{eqnarray*}}
\def\eeaa{\end{eqnarray*}}
\def\ba{\begin{array}}
\def\ea{\end{array}}
\def\be#1{\begin{equation} \label{#1}}
\def \eeq{\end{equation}}
\def\bsplit{\begin{split}}

\def\pa{\partial}

\def\dual{{\,^*}}
\def\div{\mathrm{div}\,}

\def\hot{\widehat{\otimes}}

\def\lab{\label}
\def\f12{\frac 1 2}

\def\a{{\alpha}}

\def\b{{\beta}}
\def\be{{\beta}}
\def\ga{\gamma}
\def\Ga{\Gamma}
\def\de{\delta}

\def\ka{\kappa}
\def\eps{\varepsilon}
\def\la{\lambda}

\def\si{\sigma}
\def\Si{\Sigma}
\def\om{\omega}
\def\Om{\Omega}
\def\Th{\Theta}

\def\th{\theta}
\def\vth{{\vartheta}}

\def\ka{\kappa}
\def\nab{\nabla}

\def\D{{\bf D}}

\def\H{{\bf H}}

\def\R{{\bf R}}

\def\W{{\bf W}}

\def\g{{\bf g}}

\def\k{\Th}

\def\BB{{\mathcal B}}
\def\CC{{\mathcal C}}
\def\MM{{\mathcal M}}

\def\OO{{\mathcal O}}

\def\RR{{\mathcal R}}

\def\RR{\mathcal{R}}
\def\OO{\mathcal{O}}

\def\II{{\mathcal I}}

\def\f12{{\frac 1 2}}

\def\BB{\mathcal{B}}
\def\BBd{{}^*\mathcal{B}}
\def\KK{\mathcal{K}}
\def\KKd{{}^{*}\mathcal{K}}

\def\Hb{\underline{H}}

\def\ub{\underline{u}}

\def\trch{{\mathrm{tr}}\, \chi}
\def\chih{{\widehat \chi}}
\def\chib{{\underline \chi}}
\def\chibh{{\underline{\chih}}}

\def\etab{{\underline \eta}}
\def\omb{{\underline{\om}}}
\def\bb{{\underline{\b}}}
\def\aa{\protect\underline{\a}}
\def\xib{{\underline \xi}}

\def\trchb{{\tr \,\chib}}
\def\tr{{\mathrm{tr}}\, }

\def\rhod{\, \hspace{-3pt}\dual\hspace{-2pt}\rho}

\def\Bb{\mathscr{B}}
\def\Bbd{\, ^{*  \hspace{-0.15em}}\Bb}

\def\trt{\trh \Th}

\def\Kc{\widecheck{K}}
\def\ah{\hat{a}}

\def\trth{\slashed{\mathrm{tr}}\, \th}
\def\thh{\widehat{\th}\hspace{0.05em}}
\def\thc{\widecheck{\trth}}

\def\Ricc{\mathbf{Ric}}

\def\kh{\widehat{\k}}
\def\Thh{\widehat{\Th}}

\def\0{_0}
\def\1{}

\def\ao{\widecheck{a}}

 \def\nabh{\nab\mkern-12mu /\,}

\def\laph{\slashed{\Delta}}

\def\divh{\slashed{\mathrm{div}}\hspace{0.1em}}
\def\curlh{\slashed{\mathrm{curl}}\hspace{0.1em}}
\def\trh{\slashed{\mathrm{tr}}\hspace{0.1em}}

\def\Kk{\mathscr{K}}
\def\Kkd{\, ^{*  \hspace{-0.3em}}\Kk}

\def\trch{\trh\chi}
\def\trchb{\trh\chib}

\def\divz{\slashed{\mathrm{div}}^{(0)}}
\def\curlz{\slashed{\mathrm{curl}}^{(0)}}

\def\d{\slashed{\mathcal{D}}}
\def\Kc{\widecheck{K}}

\def\ah{\hat{a}}

\def\Ricc{\mathbf{Ric}}

 \def\nabh{\nab\mkern-12mu /\,}

\def\0{^{(0)}}

\def\nabz{\nabh^{(0)\hspace{-0.1em}}}

\def\pp{p}

 \def\gz{\ga^{(0)}}
 \def\H{\mathfrak{h}}

\def\gacr{\mathring \ga}

\def\trchc{\widecheck{\trch}}
\def\trchbc{\widecheck{\trchb}}
\def\rhoc{\widecheck{\rho}}
\def\rhodc{\dual\widecheck{\rho}}
 
\begin{document}


\title{Formation of Trapped Surfaces from Spacelike Initial Data}
\author{Xuantao Chen and Sergiu Klainerman}


\date{}
\maketitle
\abstract{We make use of the free data formalism developed in \cite{CK25,CK26}  to provide a direct construction of  short-pulse type Cauchy data. The construction of the spacelike short-pulse follows from the local existence result established in \cite{CK26}.  The forward integration construction in \cite{CK26} allows us to show that such data can be extended to a set of asymptotically flat Cauchy data. This greatly extends the result of Li--Yu \cite{LiYu}.}
  
\parindent = 0 pt
\parskip = 12 pt

\tableofcontents

\section{Introduction}

In his groundbreaking work  \cite{Chr1}, Christodoulou found an important criterion for regular characteristic data,  known under the name of short pulse,   which leads through evolutions to the formation of trapped surfaces. His work was further developed in \cite{KRodn},  \cite{KLR}, 
\cite{An12}, \cite{AnLuk}, \cite{An19} etc. The defining feature of all these works is that the short pulse data is prescribed on one of two transversal characteristic hypersurfaces, as in Figure \ref{fig:finite-short-pulse} below. In that case, one can take advantage of the fact that one can easily identify the free data on each  hypersurface.\footnote{One can freely prescribe the shear on each hypersurface.}

 Prescribing regular  Cauchy  data sets $(\Si,  g, k)$,  which lead to trapped surfaces,    is  considerably more  difficult  due to the   difficulty of 
  identifying the free  Cauchy data in that case  (on which some adapted version of the short pulse can be prescribed),    i.e., the couple    $(g, k)$  which solve the constraint equations
\begin{equation}\lab{ece}
    \begin{split}
        \div k-\nab\, \tr k&=0,\\
        R_g+(\tr k)^2-|k|^2&=0.
    \end{split}
\end{equation}
Indirectly, we know that such data    must exist,   as a consequence of  the   past null infinity version of the results mentioned above,  as suggested  by    Figure \ref{fig:finite-short-pulse}.    A   more  direct  construction of the asymptotically flat Cauchy data is given by Li--Yu \cite{LiYu} (see also Li--Mei \cite{LiMei-CollapsingVacuum}). They make use of the  local spacetime region constructed in \cite{Chr1}, and apply the gluing technique of Corvino--Schoen \cite{Corvino2000,CorvinoSchoen} (see also Chruściel--Delay \cite{ChruscielDelay2003}) to obtain data that is isometric to some Kerr initial data in the exterior. We also refer to the recent advances \cite{ShenWan2025, GiorgiShenWanBosonStars, ShenWan26ADM} in the multi-trapped surfaces setting. In all these cases, however,   there is no direct  identification  of   a short pulse criterion on  a given initial   hypersurface  $\Si$.

\usetikzlibrary{arrows.meta}

\noindent

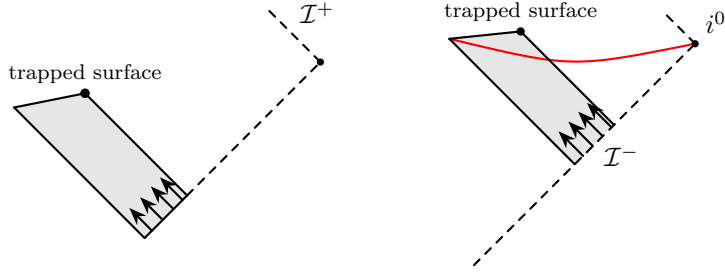
\begin{figure}[htbp]
\lab{fig1}
\centering
\begin{minipage}[t]{0.48\textwidth}
\raggedleft
\begin{tikzpicture}[
baseline=(current bounding box.north),
    scale=0.75,
    line cap=round,
    line join=round,
    >=Stealth,
    lab/.style={font=\small}
]

\coordinate (A) at (-0.55,-0.20);
\coordinate (B) at (0.20,0.55);

\coordinate (L) at (-2.85,2.10);
\coordinate (T) at (-2.10+0.5,2.85-0.5);

\coordinate (P) at (1.55+1,1.90+1);

\coordinate (IpLow) at (2.05+1,1.40+1);
\coordinate (IpHigh) at (1.15+1-0.5,2.30+1+0.5);

\fill[gray!20] (A) -- (B) -- (T) -- (L) -- cycle;

\draw[thick] (A) -- (B);
\draw[thick] (A) -- (L);
\draw[thick] (B) -- (T);
\draw[thick] (L) -- (T);

\draw[dashed, thick] (B) -- (P);

\draw[dashed, thick] (IpHigh) -- (P);
\node[lab, above] at (2.5,3.4) {$\mathcal I^+$};

\fill (P) circle (1.8pt);

\draw[->, thick] (-0.42,-0.07) -- (-0.82,0.33);
\draw[->, thick] (-0.23,0.12) -- (-0.63,0.52);
\draw[->, thick] (-0.04,0.31) -- (-0.44,0.71);
\draw[->, thick] (0.12,0.47) -- (-0.28,0.87);

\fill (T) circle (2.4pt);
\node[lab, above] at (T) {\scriptsize trapped surface};

\node[lab, below right] at (0.70,1.05) {};

\end{tikzpicture}\end{minipage}
\hfill
\begin{minipage}[t]{0.45\textwidth}
\raggedright
\begin{tikzpicture}[
baseline=(current bounding box.north),
    scale=0.65,
    line cap=round,
    line join=round,
    >=Stealth,
    lab/.style={font=\small}
]

\coordinate (A) at (-0.45+1,-0.45+1);
\coordinate (B) at (0.35+1,0.35+1);

\coordinate (L) at (-3.00+1,2.10+1);
\coordinate (T) at (-1.55+1,2.25+1);

\fill[gray!20] (A) -- (B) -- (T) -- (L) -- cycle;

\draw[dashed, thick] (-1.5,-1.5) -- (3.0,3.0);

\draw[dashed, thick] (3.0,3.0) -- (2.4,3.6);

\draw[thick,red] ({-3+1},{2.1+1}) .. controls (0.5,2.5) .. (3,3);
\fill (3.0,3.0) circle (2.1pt);
\node[lab, above right] at (3.0,3.0) {$i^0$};

\node[lab, right] at (0.95,0.72) {$\II^-$};

\draw[thick] (A) -- (L);
\draw[thick] (B) -- (T);

\draw[thick] (L) -- (T);

\draw[->, thick] (-0.28+1,-0.28+1) -- (-0.78+1,0.22+1);
\draw[->, thick] (-0.08+1,-0.08+1) -- (-0.58+1,0.42+1);
\draw[->, thick] (0.12+1,0.12+1) -- (-0.38+1,0.62+1);
\draw[->, thick] (0.30+1,0.30+1) -- (-0.20+1,0.80+1);

\fill (T) circle (2.4pt);
\node[lab, above] at (T) {\scriptsize trapped surface};

\end{tikzpicture}

\end{minipage}
\centering
\captionsetup{width=0.6\textwidth, font=small}
\caption{Penrose diagrams for finite-region short pulse (on the left) and the short pulse from past null infinity (on the right).}

\label{fig:finite-short-pulse}

\end{figure}

In this paper, we provide a  first  such direct construction\footnote{That is, we do not  appeal  to   spacetime evolution result, based on characteristic data.} of   \textit{short-pulse type}, regular,   Cauchy data  whose evolution  necessarily  contain a  trapped surface. To do that we appeal to  our previous results  \cite{CK25} and, especially  \cite{CK26}, in    the  framework of \textit{free data}     introduced in  these papers.
 More precisely:
 \begin{enumerate}
 \item  We make use of the   local existence  result  in \cite{CK26} in which the free data  $(\BB,\BBd, \KK,\KKd)$ (as identified in \cite{CK25}) are allowed to be large, consistent with  the short pulse  regime of  An--Luk \cite{AnLuk}, where an independent largeness parameter $a>0$ was  introduced.
 \item   We  extend  the local  data  by integrating  in the radial direction towards infinity,  using the global result  of  \cite{CK26}, rather than relying on gluing as in   \cite{LiYu}.
 \end{enumerate}
 
 \subsection{Review of the results in \cite{CK25,CK26}}


In \cite{CK25,CK26}, we developed a free data formalism for the Einstein constraint equations, in which we  decompose \eqref{ece} as a coupled system of transport and (2D) elliptic equations, and identify the free data $(\BB,\BBd,\KK,\KKd)$ as  four scalars supported on $\ell\geq 2$ spherical harmonic modes and verifying  (see Section \ref{secion:Notations} for the relevant  notations)
\begin{equation}\lab{eq:prescribed-scalar-conditions}
\begin{gathered}
(\divh Y-\BB)_{\ell\geq 2}=0,\quad (\curlh Y-\BBd)_{\ell\geq 2}=0, \\
\Big(\laph \big( \Pi+\frac 12 \trt \big)-\KK\Big)_{\ell\geq 2}=0,\quad \big(r^{-4} \ah N(r^4 \curlh \Xi)-\KKd\big)_{\ell\geq 2}=0.
\end{gathered}
\end{equation}
The free  scalars  have a similar physical interpretation  as    the two  shear tensors  of  the characteristic framework,  and the same number of degrees of freedom,  i.e.,   modulo gauge choices, they  exhaust the degrees of freedom of solutions to \eqref{ece}; see \cite[Section 2]{CK25} for detailed discussions.

Meanwhile, in \cite{CK25,CK26}, we deal with two different scenarios:
\begin{itemize}
\item Integration from infinity in \cite{CK25}: We prescribe the free scalar with relatively fast decaying rate, as well as the ADM parameters at infinity, and provides a construction of data in an exterior region.\footnote{To make it a complete Cauchy data on $\mathbb{R}^3$, one can apply the fill-in result by Bieri--Chru\'sciel \cite{BieriChrusciel2017}
 or Mao--Oh--Tao \cite{MaoOhTao2023}.} Such data can decay at the slowest $O(r^{-1-\de})$ at the metric level (subtracting the mass tail), and can be aribitraily fast.
\item Integration towards infinity in \cite{CK26}: Given the data on a finite sphere, we prescribe the free scalars on its exterior side. We then integrate in the forward direction of $r$ to obtain the solution to \eqref{ece}. This produces data with slow decay rates, in particular those with borderline decay studied in \cite{Shen23, Shen24}.
\end{itemize}
The result in \cite{CK26} combined a   local existence result, based on an iteration procedure,  with a global  extension  argument.  
We proved in fact a stronger local existence  result than needed in the context of  that paper, one in which  the free scalars $(\BB,\BBd,\KK,\KKd)$ can have arbitrarily large size in terms of  a  parameter   $a>0$.  As it turns out, this  fits well with  the  scale-critical framework of An--Luk \cite{AnLuk}.


\subsection{Statement of the result}

\begin{theorem}[Main theorem]\lab{thm:short-pulse-data}
Let $a>0$ be a universal large constant. Consider the unit ball $\{|x|\le 1$\} with trivial data.\footnote{This assumption can be easily relaxed.}
There exists a small constant $\eps_0>0$, independent of $a$, 
such that given two sets of free scalars $(\BB_{sp},\BBd_{sp},\KK_{sp},\KKd_{sp})$ and $(\BB_{ext},\BBd_{ext},\KK_{ext},\KKd_{ext})$, defined respectively on $(1,2)\times \mathbb{S}^2$ and $(1,\infty)\times \mathbb{S}^2$ and satisfying, for some $s\geq 6$ and $\si \in [0,1)$, with the norms defined in Definition \ref{def:L2-Hs-S_r},
\begin{equation*}
\begin{gathered}
\sup_{r\in [1,2]} \| (\BB_{sp},\BBd_{sp},\KK_{sp},\KKd_{sp})\|_{\H^s(S_r)} \leq a^\frac 12,\\
\sup_{r\in (1,\infty)} r^{2+\si} \| (\BB_{ext},\BBd_{ext},\KK_{ext},\KKd_{ext}) \|_{\H^s(S_r)} \leq \eps_0,
\end{gathered}
\end{equation*}
the following statements hold:
\begin{enumerate}
\item\lab{statement:existence-data}
 There exists a complete Cauchy data $(g,k)$ on $\mathbb{R}^3$ solving \eqref{ece}, matching  the data inside the unit ball,  and, for some $\de>0$, verifies \eqref{eq:prescribed-scalar-conditions} for $(\BB_{sp},\BBd_{sp},\KK_{sp},\KKd_{sp})$ on $(1,1+\de)\times \mathbb{S}^2$, and for $(\BB_{ext},\BBd_{ext},\KK_{ext},\KKd_{ext})$ on $(1+\de,\infty)\times \mathbb{S}^2$. 
 \item\lab{statement:semi-global-existence}
   As in the familiar characteristic case, our data  can be  evolved into a  spacetime region $\MM_{\CC}$, by      a  semi-global existence. 
\item\lab{statement:trapped-surface}
Moreover, if the lower bound condition \eqref{eq:lower-bound-condition-free-scalars} is verified, then  the future Cauchy development of $(\Si,g,k)$ must contain a trapped surface.
\end{enumerate}
\end{theorem}

We will prove statement \ref{statement:existence-data} in Section \ref{sec:construction-Cauchy-data}, and statements \ref{statement:semi-global-existence} and \ref{statement:trapped-surface} in Section \ref{sec:spacetime-evolution}.  
\begin{remark}
Note that all four free scalars in the short-pulse region are allowed to be large. In the null spacetime language, this means that both the incoming and outgoing gravitational waves may be strong; see Figure \ref{figure:triangular region}. This shows that characteristic short-pulse data can arise from a larger class of spacelike data, without imposing a hierarchy in which the incoming radiation is dominant. In particular, our result includes a family of time-symmetric data, parametrized by $(\BB_{sp},\BBd_{sp})$, which is not accessible through the previous indirect methods.
\end{remark}

\begin{figure}[htbp]
    \centering
\begin{tikzpicture}[
    scale=2.5,
    >=Stealth,
    line cap=round,
    line join=round
]
\fill[green!10]
    (1,0) -- (0,1) -- (0.3,1.1) -- (1.2,0.2) -- cycle;
\fill[blue!10]
    (1,0) -- (1.2,0.2) -- (1.4,0) -- cycle;
\draw[thick]
    (1,0) -- (0,1) -- (0.3,1.1) -- (1.4,0) -- cycle;
        \filldraw (1.2,0) node[below,font=\scriptsize,text width=1.2cm,align=center] {short pulse};
        
\draw[->, thick] (1.08,0.02) -- (0.96,0.14);
\draw[->, thick] (1.16,0.02) -- (1.04,0.14);

\draw[->, thick] (1.24,0.02) -- (1.36,0.14);
\draw[->, thick] (1.32,0.02) -- (1.44,0.14);
        \draw[dashed] (0,0) -- (1,0) node[midway,below,font=\scriptsize,text width=1.7cm,align=center] {trivial data (or perturbation)};

\draw[thick] (1.2,0) -- (2.5,0)
node[midway,below,font=\scriptsize,text width=1.7cm,align=center] {Forward integration};
  \definecolor{paperbg}{RGB}{250,248,240}
\filldraw[fill=paperbg, draw=black] (2.5, 0) circle (0.6pt);
\node[below] at (2.5, 0) {\scriptsize $i^0$};

\node[above] at (0.6, 0.5) {\scriptsize $\MM_{\CC}$};
  \filldraw[black] (0.3,1.1) circle (0.6pt) node[above] {\scriptsize trapped surface};

      \coordinate (p) at (1.2,0.1);

\node[anchor=west, inner sep=0pt] (lab) at (1.5,0.5) {\scriptsize $\MM_{\triangle}$};
\draw[->] ([xshift=0.5pt,yshift=-0.5pt]lab.south west) to[out=220,in=40] (p);
      
    \end{tikzpicture}
    \centering
    \captionsetup{width=0.6\textwidth, font=small}
    \caption{Illustration of Theorem \ref{thm:short-pulse-data}.}\lab{figure:triangular region}
\end{figure}
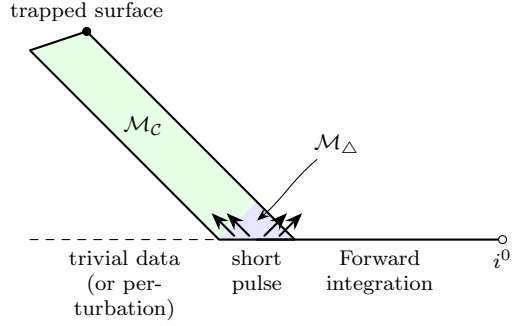

\begin{remark}[Comparison with Li--Yu \cite{LiYu}]
The result of Li--Yu \cite{LiYu}  is based on gluing the interior data to a Kerr exterior. The main obstruction is  given by  the $\ell\leq 1$ modes, 
due to which   one needs to pick the correct ADM parameters,  of  the exterior Kerr family, to  implement the gluing.\footnote{This can be thought of as a special case of the backward construction in our previous work \cite{CK25}. Indeed, the new ingredient we add in \cite{CK25} is that we can also freely add $\ell\geq 2$ parts to the data through the free scalars.} 
In contrast, the forward construction\footnote{ We note, however,    that the $\ell\leq 1$   modes  play an important role in the   construction of solutions in \cite{CK25,CK26}.}  in this paper directly extends the interior data to spatial infinity. 
\end{remark}

\subsection{Acknowledgements}
The authors thank Xinliang An for helpful comments that partly motivated the inclusion of the appendix. 
The first author is supported by ERC-2023 AdG 101141855 BlaHSt. The second author is supported by the NSF grant 2453843. 

\section{Notations, Preliminaries}
\lab{secion:Notations}

\subsection{Geometric quantities introduced in \cite{CK25}}\lab{sec:geometric-quantities-Si}

We recall the notation of geometric quantities introduced in \cite{CK25}. Assume that $\Si:=(1,\infty)\times \mathbb{S}^2$ is equipped with a metric $g$ and the associated Levi--Civita connection $\nab$, and consider the foliation by the level spheres $S_r=\{r=\mathrm{const}\}$. 
Let $N$ be the unit normal to $S_r$ in $\Si$. 
We then choose an orthonormal frame $\{e_1,e_2\}$ tangent to $S_r$, so that $\{N,e_1,e_2\}$ forms an orthonormal frame on $\Si$.

  {\bf Induced metric.}   
The induced metric on  $S_r$ is denoted by $\ga$, with the associated Levi--Civita connection $\nabh$.         
The Gauss curvature of the $r$-spheres is denoted by $K=K_\ga$.     Given standard spherical coordinates $(\vth^A)$, we also denote $\gacr$ as the standard unit round sphere metric in $(\vth^A)$.

{\bf The lapse function.}
The lapse function of a given $r$-foliation is defined as $\ah:=|\nab r|_g^{-1}$. It is independent of the choice of the angular coordinates $(\vth^1,\vth^2)$.

We define the following quantities
 \begin{equation*}
 \begin{gathered}
 \pp_a:=g(\nab_N N,e_a),\quad  \th_{ab}:=g(\nab_a N, e_b),\\
\slashed{R}_{ab}:=R(N,e_a,N,e_b),\quad Y_a:=R(N, e_b, e_b, e_a), \\
\Th_{ab}:=k(e_a,e_b),\quad \Xi_a:=k(N,e_a),\quad \Pi :=k(N,N).
\end{gathered}
\end{equation*}
 The (sphere) trace and the traceless part of $\th$ are denoted respectively by $\trth$ and $\thh$.  The sphere trace and the traceless part of $\k$ are denoted respectively by $\trh \k$ and $\kh$.

\begin{lemma}\lab{lem:loga}
We have 
\bea\lab{eq:P-loga}
\pp=-\slashed{d} (\log\ah),
\eea
where $\slashed{d}$ denotes the differential on      $r$-level surfaces.
\end{lemma}
\begin{proof}
See Lemma 2.2 in \cite{CK25}.
\end{proof}

{\bf Gauge scalars.}   As in \cite{CK25}, we define
\begin{align}
\lab{eq:def-mu}
\mu& :=-\slashed{\Delta}(\log\ah)+K-\frac 14(\trth)^2, \\
\lab{eq:def-nu}
\nu & :=\divh\Xi=\de^{ab} \nabh_a \Xi_b,
\end{align}
where we recall the following definitions for two spherical tangent $1$-forms $\psi$, $\phi$, with $\in_{ab}$ denoting the area form oriented by $\in_{12}\, =1$:
\begin{equation*}
\begin{gathered}
    \psi\cdot \phi:=\delta^{ab}\psi_a\phi_b,\quad \psi\wedge\phi:=\in^{ab} \psi_a \phi_b,\quad (\psi\hot\phi)_{ab}=\psi_a\phi_b+\psi_b\phi_a-\delta_{ab}\psi\cdot\phi,\\
    \divh \psi:=\delta^{ab}\, \nabh_a \psi_b,\quad \curlh \psi:=\, \in^{ab} \nabh_a \psi_b,\quad (\nabh\hot \psi)_{ab}:=\nabh_a \psi_b+\nabh_b \psi_a-\delta_{ab}\, \divh\psi.
    \end{gathered}
\end{equation*}

{\bf Linearized quantities.} We define the following linearized quantities
\begin{equation}\lab{eq:linearized-quantities}
\thc:= \trth-2r^{-1},\quad \Kc:=K-r^{-2},\quad \ao:= \ah-1.
\end{equation}

{\bf Norms.} We have the following definition of the weighted Sobolev norms on $S_r$ spheres.
\begin{definition}\lab{def:L2-Hs-S_r}
For $S_r$-tangent covariant rank-$k$ tensors $U_{a_1\cdots a_k}$,
we denote by $L_{\ga}^2(S_r)$ the $L^2$ space with respect to the metric $\ga$, and by $\H_{\ga}^s(S_r)$ the Sobolev spaces for positive integers $s$, defined through $r\nabh$, i.e.,
\begin{equation*}
\| U \|_{\H_\ga^s(S_r)}:= \sum_{i\leq s} \|(r\nabh)^i U \|_{L_\ga^2(S_r)}.
\end{equation*}
We also denote for simplicity $\H_0^s:=\H_{\gz}^s$,  with  $\gz:= r^2 \gacr$.
\end{definition}

\subsection{Spacetime quantities}
Consider a spacetime with a choice of null frame $\{e_3,e_4,e_a\}_{a=1,2}$ at each point, where $e_3$ and $e_4$ are transveral null pairs with $\g(e_3,e_4)=-2$, and $\{e_a\}_{a=1,2}$ are orthgonal to both $e_3$ and $e_4$. 
We recall the spacetime geometric quantities:
\begin{equation*}
    \chi_{ab}=\g(\D_a e_4,e_b),\quad \chib_{ab}=\g(\D_a e_3,e_b),      \quad  \eta_a=\frac 12 \g(\D_3 e_4,e_a), 
\end{equation*}
\begin{equation*}
 \etab_a=\frac 12 \g(\D_4 e_3,e_a), \quad \zeta_a=\frac 12 \g(\D_a e_4,e_3),\quad \om=\frac 14 \g(\D_4 e_4,e_3),
\end{equation*}
\begin{equation*}
 \omb=\frac 14 \g(\D_3 e_3,e_4),\quad  \xi_a=\frac 12 \g(\D_4 e_4,e_a), \quad \xib_a=\frac 12 \g(\D_3 e_3,e_a),
\end{equation*}
\begin{equation*}
    \a_{ab}=\W_{a4b4},\quad \b_a=\frac 12 \W_{a434},\quad \rho=\frac 14 \W_{3434},
    \end{equation*}
    \begin{equation*}
     \dual\rho=\frac 14\dual \W_{3434},\quad \bb_a=\frac 12 \W_{a334},\quad \underline \a_{ab}=\W_{a3b3}.
\end{equation*}
Here $\W$ is the Weyl tensor that can be expressed as, with $\R$ the spacetime Riemann curvature tensor,
\begin{equation}\lab{eq:Weyl-tensor}
\bsplit
\W_{\rho\sigma\mu\de} &=  \R_{\rho\sigma\mu\de}
	+	\frac{1}{2} \left( \g_{\rho\mu} \Ricc_{\de\sigma} - \g_{\rho\de} \Ricc_{\mu\sigma} - \g_{\sigma\mu} \Ricc_{\de\rho} + \g_{\sigma\de} \Ricc_{\mu\rho} \right) \\
	& \quad + \frac{1}{6} \R_g \left( \g_{\rho\mu} \g_{\de\sigma} - \g_{\rho\de} \g_{\mu\sigma} \right).
	\end{split}
\end{equation}
We also define the following renormalized curvature quantities, first introduced in \cite{KR05}:
\begin{equation}\lab{def:rhoc-rhodc}
\widecheck{\rho}:=\rho-\frac 12\chih\cdot\chibh,\quad \dual\widecheck{\rho}:=\rhod-\frac 12\chih\wedge\chibh.
\end{equation}
They verify the improved propagation  equations\footnote{Note the  absence of $\a, \aa$ in those equations.\lab{foot;a-aa}} \eqref{eq:nab3P} and \eqref{eq:nab4bb}  and as such  they   played an important  role in many works, including \cite{LR15}, \cite{LR17}, \cite{AnLuk}, \cite{Shen23}.

\begin{proposition}\lab{prop:b-bb-expression}
For $\Si$ embedded in a spacetime $(\MM,\g)$ with a specified $r$-foliation,  the following relations hold true between 
the  intrinsic quantities   defined in Section \ref{sec:geometric-quantities-Si} and the spacetime quantities  defined above:
\begin{align}
\lab{eq:b+bb}
(\b+\bb)_a &= 2\left(Y+\Xi \cdot \k-\trt \Xi\right)_a+3\Ricc_{Na},\\
\lab{eq:b-bb}
(\b-\bb)_a &= -2\left(\nabh \Pi-\nabh_N \Xi-2\th\cdot \Xi+p\cdot  \Th-\Pi p \right)_a+(\slashed{\CC}_{Mom})_a,\\
\lab{eq:Gauss-rho}
\rho &= -K -\frac 14 (\trt)^2+\frac 14 (\trth)^2 +\frac 12 |\kh|^2-\frac 12 |\thh|^2+\frac 12 \CC_{Ham} \\
\nonumber & \quad -\Big(\Ricc-\frac 12 (\R_{\g}) \g\Big)_{NN}+\frac 12 \Big(\Ricc-\frac 12 (\R_{\g}) \g\Big)_{aa} -\frac 23 \R_{\g},\\
\lab{eq:curl-kn-dual-rho}
\dual\rho&= -\curlh \Xi- \kh\wedge\thh.
\end{align}
\end{proposition}
\begin{proof}
See Proposition 2.19 in \cite{CK25}.
\end{proof}

\subsection{The null  structure and null  Bianchi equations}\lab{sec:null-str-Bianchi-equations}
We recall the null structure and Bianchi equations first derived in \cite{C-K}. The following equations are their special form in the double null foliation, see more details in Section \ref{sec:setup-Si-doublenull}.
\begin{align}
\lab{eq:e4Raych}
\nabh_4 \trch&=-\frac 12(\trch)^2 -|\chih|^2-2\om \trch,\\
\lab{eq:Codazzi-chih}
\divh \chih &= \frac 12 \nabh\trch-\frac 12 (\eta-\etab)\cdot (\chih-\frac 12\trch)-\b,\\
\lab{eq:e3Raych}
\nabh_3 \trchb &= -\frac 12 (\trchb)^2 -|\chibh|^2 -2\omb\trchb,\\
\lab{eq:Codazzi-chibh}
\divh \chibh &= \frac 12 \nabh\trchb+\frac 12 (\eta-\etab)\cdot (\chibh-\frac 12\trchb)+\bb,
\end{align}
\begin{align}
\nabh_4 \eta &=-\chi\cdot (\eta-\etab)-\b,\\
\nabh_3 \etab &=-\chib\cdot (\etab-\eta)+\bb,\\
\nabh_4 \omb &=2\om\omb+\frac 34 |\eta-\etab|^2-\frac 14 (\eta-\etab)\cdot (\eta+\etab)-\frac 18 |\eta+\etab|^2+\frac 12 \rho,\\
\nabh_3\om &=2\om \omb+\frac 34|\eta-\etab|^2+\frac 14 (\eta-\etab)\cdot (\eta+\etab)-\frac 18 |\eta+\etab|^2 +\frac 12 \rho,
\end{align}
\begin{align}
\lab{eq:nab3b}
\nabh_3 \b&= \nabh\rhoc+\dual\nabh\rhodc+(2\omb-\trchb)\b+2\chih\cdot\bb+3(\eta\rhoc+\dual\eta\rhodc)\\
&\notag \quad +\frac 12 (\nabh(\chih\cdot\chibh)+\dual\nabh(\chih\wedge\chibh))+\frac 32 (\eta\chih\cdot\chibh+\dual\eta\chih\wedge\chibh), \\
\lab{eq:nab4P}
\nabh_4(\rhoc,\rhodc) &= (\divh\b,-\curlh\b)-\frac 32 \trch (\rhoc,\rhodc)+((\zeta+2\etab)\cdot\b,-(\zeta+2\etab)\wedge\b)\\
&\notag \quad +\left(-\frac 12 \chih\cdot (\nabh\hot\etab+\etab\hot\etab),-\frac 12 \chih\wedge(\nabh\hot\etab+\etab\hot\etab)\right),\\
\lab{eq:nab3P}
\nabh_3 (\rhoc,\rhodc) &= -(\divh\bb,\curlh\bb)-\frac 32 \trchb (\rhoc,\rhodc)+((\zeta-2\eta)\cdot\bb,(\zeta-2\eta)\wedge\bb)\\
&\notag \quad +\left(-\frac 12 \chibh\cdot (\nabh\hot\eta+\eta\hot\eta),\frac 12 \chibh\wedge(\nabh\hot\eta+\eta\hot\eta)\right),\\
\lab{eq:nab4bb}
\nabh_4 \bb &= -\nabh\rhoc+\dual\nabh\rhodc+(2\om-\trch)\bb+2\chibh\cdot \b+3(-\etab \rhoc+\dual\etab\rhodc) \\
&\notag \quad -\frac 12 (\nabh (\chih\cdot\chibh)-\dual\nabh(\chih\wedge\chibh))-\frac 32(\etab\chih\cdot\chibh-\dual\etab \chih\wedge\chibh).
\end{align}
Note that we consider the renormalized quantities $\rhoc$ and $\rhodc$ instead of $\rho$ and $\rhod$. As a result, the quantities $\a$ and $\aa$ are decoupled from the system.

\section{Construction of the Cauchy data}\lab{sec:construction-Cauchy-data}

\subsection{The short pulse}

In this section, we provide a direct construction of the short-pulse type Cauchy data. We recall the local existence result obtained in \cite{CK26}:
\begin{definition}\lab{def:almost-round-sphere-data}
We refer to the following as a set of canonical sphere data of radius $r_0$:
\begin{itemize}
\item A Riemannian $2$-sphere $(S,\ga)$, with $\ga=r_0^2 e^{2u} \gacr$ in some coordinates $(\vth^A)$;
\item The radial expansion $\trth$ and the expansion in the time direction $\trt$;
\item The $1$-form $\Xi$ on $(S,\ga)$ with $\divh \Xi=0$.
\end{itemize}
For a given nonnegative integer $i$, we say a set of canonical sphere data $(S,u,\vth^A, \trth,\trt, \Xi)$ is $O_{i}(\eps_0)$-almost round with radius $r_0$, if the data satisfies the estimate
\beaa
r_0^{-1} \| u \|_{\H_0^{i}(S)} \lesssim \eps_0,\quad \| \trth-2r_0^{-1}, \trt,\Xi \|_{\H_0^{i-1}(S)} \lesssim \eps_0. 
\eeaa
Here, the $\H_0^i$ norms are defined using the round metric in the coordinates $(\vth^A)$ as in Definition \ref{def:L2-Hs-S_r}.
\end{definition}

\begin{theorem}[Local existence, {\cite[Theorem 3.4]{CK26}}]\lab{thm:local-existence-precise}
Given $s\geq 3$, consider a set of $O_{s+2}(\eps_0)$-almost round canonical sphere data $(S,u,\vth^A,\trth,\trt,\Xi)$ of radius $r_0$, as in Definition~\ref{def:almost-round-sphere-data}, with smallness parameter $\eps_0 \ll 1$.
Then, given four scalars $(\BB,\BBd, \KK,\KKd)$ on $(r_0,r_0+1)\times \mathbb{S}^2$,\footnote{As the statement is entirely local in $r$, it would suffice to prescribe them only on a sufficiently small interval near $r_0$.}      there exists
$\eps_{loc}>0$, depending on $r_0$ and $\sup_r \|\BB,\BBd, \KK,\KKd\|_{\H^s(S_r)}$, and a set of initial data $(g,k)$ on $(r_0,r_0+\eps_{loc})\times \mathbb{S}^2$ solving \eqref{ece}, such that the metric can be written as
\begin{equation}\lab{eq:metric-abu}
g=\ah^2 dr^2+ r^2 e^{2u} \gacr_{AB} (d\vth^A +b^A dr)(d\vth^B+b^B dr),
\end{equation}
and the data satisfies the following:
\begin{itemize}
\item The gauge conditions
\begin{itemize}
\item The $\ell=0$ conditions
\begin{equation}
\lab{eq:ell=1-condition-u-solution1}
\overline{\ao}=0,\quad \overline{\Pi}=-\frac 12 \overline{\trt}. 
\end{equation}
\item The $\ell=1$ conditions 
\begin{equation}\lab{eq:ell=1-condition-u-solution}
\int_{\mathbb{S}^2} e^{2u} x^i =0,\quad (e^{2u} \curlh b)_{\ell=1}=0;
\end{equation}
\item The conditions on the gauge scalars defined in \eqref{eq:def-mu}-\eqref{eq:def-nu}
\begin{equation}
\lab{eq:ell=1-condition-u-solution3}
\mu_{\ell\geq 1}=\nu=0.
\end{equation}
\end{itemize}
\item The free scalar conditions \eqref{eq:prescribed-scalar-conditions}.
\end{itemize}
Moreover, for any positive $\eps_{loc}'<\eps_{loc}$ and a set of initial data $(g',k')$ defined on $(r_0,r_0+\eps_{loc}')\times \mathbb{S}^2$ solving \eqref{ece} and verifying    the same conditions  as $(g, k)$, we have $(g,k)|_{(r_0,r_0+\eps_{loc}')\times \mathbb{S}^2}=(g',k')$.
\end{theorem}

Note that the theorem only ensures the local existence, i.e., $\eps_{loc}$ should be viewed as a very small number. In particular, one should expect $\eps_{loc} a\ll 1$, where 
\begin{equation}\lab{eq:def-a}
a:= \sup_r \| \BB_{sp},\BBd_{sp}, \KK_{sp},\KKd_{sp} \|^2_{\H_0^s(S_r)}.
\end{equation}
Recall that the result of An--Luk \cite{AnLuk} and many subsequent works allows the scale-critical result, i.e., $\de:=\eps_{loc}$ can be arbitrarily small and independent of $a$. Our local existence result fits well into this regime. To cover the original case $a\approx \de^{-1}$ of Christodoulou, we need to make some modifications, see Appendix \ref{appendix:a-de-inverse}. Here we focus on the case where one has $\de a\ll 1$.

\begin{proposition}\lab{prop:bounds-local-existence}
For $(\BB_{sp},\BBd_{sp},\KK_{sp},\KKd_{sp})$ and $a>0$ defined in \eqref{eq:def-a}, the solution obtained in Theorem \ref{thm:local-existence-precise} also satisfies the following estimates:
\begin{equation*}
\begin{gathered}
\sup_{r\in (1,1+\de)} \| \Kc \|_{\H_0^s(S_r)} \lesssim \eps_0+\de a^\frac 12,\quad \sup_{r\in (1,1+\de)} \| \trth, \trt \|_{\H_0^{s+1}(S_r)}  \lesssim \eps_0+\de a, \\
\quad \sup_{r\in (1,1+\de)} \| \Xi ,\slashed{d} (\log\ah) \|_{\H_0^{s+1}(S_r)}  \lesssim \eps_0+\de a^\frac 12, \quad
\sup_{r\in (1,1+\de)} \| \thh,\Thh, \Pi, Y \|_{\H_0^s(S_r)}  \lesssim \eps_0+a^\frac 12.
\end{gathered}
\end{equation*}
\end{proposition}
\begin{proof}
This follows from re-examining the proof of Theorem \ref{thm:local-existence-precise}, which was given in \cite{CK26}, with $\de=\eps_{loc}$ and $a=M^2$ in the proof there. Indeed, through the iteration, we obtain the bounds of all geometric quantities. For instance, the equations for $\thc$, $\Kc$, and $\trt$ reads, respectively,
\begin{align}
\lab{eq:thc}
N \thc &= \mu -2r^{-1}\thc-2 r^{-2}\ao +\widetilde\Ga_1\cdot \widetilde\Ga_1, \\
\lab{eq:K}
N \Kc &= r^{-1} \mu- \divh Y-3r^{-1}\Kc -2 r^{-3} \ao+\Ga_1\cdot \Ga_2\\
&\notag \quad +2\, \slashed{d}(\log\ah)\cdot Y-\thh\cdot \left(\nabh\hot\, \slashed{d}(\log\ah)-\slashed{d}(\log\ah)\hot\, \slashed{d}(\log\ah)\right) ,\\
\lab{eq:trt}
N \trt&= \divh \Xi + 2 r^{-1} \Pi  - r^{-1} \trt +\widetilde\Ga_1\cdot \widetilde\Ga_1,
\end{align}
where the size of $\widetilde \Ga_1$ is at its worst $\sim a^\frac 12$ (in \cite{CK26} we denote $M=a^\frac 12$). This implies the estimates for $\thc$ and $\trt$.
In \cite{CK26} we obtained
\beaa
\| \Kc, \curlh\Xi \|_{\H_0^s(S_r)} \lesssim  \eps_0+\de (a^\frac 12)^2.
\eeaa
 This is not optimal in terms of the power of $a$; in fact, simply examining the right-hand side of the equation of $\Kc$ and $\curlh\Xi$, we deduce
\beaa
\| \Kc , \curlh \Xi \|_{\H_0^s(S_r)} \lesssim \eps_0+\de a^\frac 12 ,
\eeaa
and this yields the estimate as required.\footnote{We have kept the discussion brief; the more detailed bookkeeping, which also covers the case $\de a\sim 1$, is given in Appendix \ref{appendix:a-de-inverse}.}
\end{proof}

\begin{remark}
Note that in our construction, the derivatives of $\thh$ and $\Thh$ in the direction of $N$ are never used. In other words, one can allow $\thh$ and $\Thh$ to jump in the radial direction. This aligns with the observation in the characteristic setting that $\chih$ is allowed to jump in the outgoing null direction.\footnote{In fact, as we have already seen from the null structure equation, the curvature component $\a$, which is roughly the $e_4$ derivative of $\chih$, decouples from the system; see  footnote \ref{foot;a-aa}.} In our situation, one can in fact have $\thh$ and $\Thh$ jump at the same point, resembling an interaction of two transversal impulsive waves. In the characteristic setting, this is constructed by two transversal impulsive waves meeting each other in Luk--Rodnianski \cite{LR17}.
\end{remark}

\subsection{Extension to asymptotically flat data}
We recall the extension result in \cite{CK26}.\footnote{In \cite{CK26}, we state the result for the borderline decay, but pointed out in Remark 4.4 there that the proof works for all the rates stated here.}
\begin{theorem}[{\cite[Theorem 4.3]{CK26}}]
\lab{main-thm-CK26}
Fix a positive integer $s\geq 3$ and a small constant $\eps_0>0$. For given $O_{s+2}(\eps_0)$-almost round canonical sphere data $(S,u,\vth^A,\trth,\trt,\Xi)$ of radius $r_0$, as defined in Definition \ref{def:almost-round-sphere-data}, and prescribed scalars $(\BB,\BBd, \KK,\KKd)$ on $(r_0,\infty)\times \mathbb{S}^2$ satisfying the bounds
\begin{equation*}
\sup_{r\in (r_0,\infty)} r^{2+\si} \| \BB,\BBd,\KK,\KKd \|_{\H_0^s(S_r)} \leq \eps_0 ,
\end{equation*}
there exists a unique solution to \eqref{ece} on $(r_0,\infty)\times \mathbb{S}^2$ satisfying the conditions listed in  the local existence result of Theorem \ref{thm:local-existence-precise}. Moreover, for quantities defined in Section \ref{sec:geometric-quantities-Si}, we have the estimate
\begin{equation*}
|\ao, b, u|\lesssim \eps_0 r^{-\si},\quad |\thh, \thc, \trt, \Thh, \Xi, \Pi |\lesssim \eps_0 r^{-1-\si},\quad |\Kc, Y|\lesssim \eps_0 r^{-2-\si}.
\end{equation*}
\end{theorem}

To establish the global(-in-$r$) existence, it suffices to verify the sphere data on $S_{1+\de}$, in view of Theorem \ref{main-thm-CK26}. We have obtained in Proposition \ref{prop:bounds-local-existence} the following bounds
\begin{equation*}
\Kc \sim \de a^\frac 12,\quad \trth, \trt \sim \de a,\quad \Xi \sim \de a^\frac 12.
\end{equation*}
Therefore, taking $\de$ such that $\eps_0:=\de a \ll 1$, we can apply Theorem \ref{main-thm-CK26}, using the scalars $(\BB_{ext},\BBd_{ext},\KK_{ext},\KKd_{ext})$, to obtain the asymptotically flat Cauchy data as stated.

\begin{remark}
[Extension to fast decaying data]
Note that the construction above only yields the decay
$g-\de=O(r^{-\si})$, $0<\si<1$. Without further asymptotic
cancellations, such a decay rate is not consistent with finite ADM charges. This contrasts with the data constructed by the gluing method
in \cite{LiYu}, for which the ADM charges are finite.
The obstruction can already be seen from forward integration in a
scalar model:
\begin{equation}\lab{eq:scalar-model}
\pa_r f+\la r^{-1} f=F .
\end{equation}
Then
\[f(r)=\left(\frac{r_0}{r}\right)^\la f(r_0)
+r^{-\la}\int_{r_0}^r (r')^\la F(r')\,dr' . \]
Thus the forward construction naturally leads to an $r^{-\la}$
homogeneous contribution. When the source decays sufficiently fast,
improving upon this decay requires a more sophisticated modulation argument. 

In the forthcoming work \cite{Gauge-scalars-preparation}, we construct
data with faster decay and relate this issue to the choice of gauge
scalars. More precisely, we show that an appropriate choice of the gauge
scalars makes the effective coefficient $\la$ in
\eqref{eq:scalar-model} favorable for the $\ell\geq 2$ components.
Consequently, the remaining obstructions are confined to the
$\ell\leq 1$ modes, in agreement with the finite-dimensional
obstructions appearing in the gluing construction.
\end{remark}

\section{The spacetime evolution}\lab{sec:spacetime-evolution}
We prove  Parts  \ref{statement:semi-global-existence} and \ref{statement:trapped-surface} of Theorem \ref{thm:short-pulse-data}   by  appealing to the corresponding   characteristic  data  results   in \cite{Chr1}, \cite{KRodn},  \cite{KLR}, \cite{An12}, \cite{AnLuk}, \cite{An19}. We follow, in fact,  our own framework in  \cite{CK-TS}. The strategy is very simple:    we extend the spacetime to $H_{-1}$, defined as the future outgoing null cone of $S_1$,  and show that the induced  characteristic data on $H_{-1}$  verify the assumptions  needed in   \cite{CK-TS}. We  distinguish  between the  small  triangular region $\MM_{\triangle}$---the domain of dependence  of  our  short pulse region---and  the long,  semi-global, ingoing region $\MM_{\CC}$   determined   from    $H_{-1}$ (see Figure \ref{figure:triangular region}).  In contrast to the geodesic foliation used in \cite{CK-TS}, we construct the spacetime region $\MM_{\triangle}$ using a double-null foliation.    This is due to the fact that, initially, in the  region  $\MM_{\triangle}$ both  $\nabh_3$ and $\nabh_4$ behave symmetrically,   like $\de^{-1}$ (and  $\nabh\sim 1$),    while in $\MM_{\CC}$    the symmetry breaks down.

\subsection{Set up of spacetime quantities on \texorpdfstring{$\Si$}{Si}}
\lab{sec:setup-Si-doublenull}

To set up the double null foliation, we initialize the outgoing and incoming optical function $u$ and $\ub$ on $\Si$ by $u=-r$, $\ub=r$. We first recall the null lapse function and the choice of the null pair in the double null foliation:
\beaa
\frac 12 \Om^{-2} = -\g (\grad\, u,\grad\, \ub),\quad e_3 := -2\Om\, \grad\, \ub,\quad e_4= -2\Om\, \grad\, u.
\eeaa
On $\Si$, we can use the null pair $T+N$, $T-N$ to compute
\beaa
\frac 12 \Om^{-2} = -\g (\grad\, u,\grad\, \ub)=\frac 12 (T-N)(u) (T+N)(\ub).
\eeaa
Note that $\{T-N,T+N\}$ determines, at each $r$-sphere on $\Si$, the two null directions orthogonal to the sphere. Therefore, $\{e_3,e_4\}$ is a rescale of $\{T-N,T+N\}$ everywhere on $\Si$. As a consequence, $(T-N)(\ub)|_\Si=0$, $(T+N)(u)|_{\Si}=0$. Hence we have
\beaa
\frac 12 \Om^{-2} = -2N(u) N(\ub)=2(N(r))^2=2\ah^{-2},\quad \text{hence $\Om=\frac 12 \ah$ on $\Si$.}
\eeaa
On the other hand, 
\beaa
\g(e_3,N)= -2\Om N(\ub)= -2\Om\, \ah^{-1} =-1.
\eeaa
Therefore, $e_3=T-N$ and $e_4=T+N$ on $\Si$. We then have the following spacetime Ricci coefficients determined on $\Si$:
\begin{equation*}
\chi_{ab}:=\g(\D_a e_4, e_b)=\th_{ab}+\Th_{ab}
, \quad \chib_{ab}:=\g(\D_a e_3, e_b)=\th_{ab}-\Th_{ab},\quad  \zeta_a:=\frac 12 \g(\D_a e_4,e_3)=-\Xi_a.
\end{equation*}
It also holds in the double null gauge that
\beaa
\eta_a=\zeta_a+\nabh_a(\log\Om),\quad \etab_a=-\zeta_a+\nabh_a(\log\Om), \quad \xi=\xib=0.
\eeaa
Since we also know $\nabh\log\Om$, we determine $\eta$, $\etab$ on $\Si$. 

To determine $\om$ and $\omb$, we recall the identities in the double null foliation $\om=-\frac 12 \nabh_4(\log\Om)$, $\omb=-\frac 12 \nabh_3(\log\Om)$. Therefore, on $\Si$, we have
 \begin{equation*}
 \om-\omb=-\nabh_N (\log\Om)=-\nabh_N (\log\ah),
\end{equation*}
On the other hand,
 \begin{equation*}
\om=\frac 14 \g(\D_4 e_4, e_3)= \frac 14 \g(\D_3 e_4,e_3) +\frac 14\g(\D_{2N} (T+N),T-N) = -\omb+\g(\D_N N,T)=-\omb-\Pi,
\end{equation*}
i.e., $\om+\omb=-\Pi$. Therefore,
 \begin{equation*}
\om=-\frac 12(\nabh_N (\log\ah)+\Pi),\quad \omb=\frac 12(\nabh_N (\log\ah)-\Pi).
\end{equation*} 
According to Proposition \ref{prop:b-bb-expression}, we have, in an Einstein-vacuum spacetime,
 \begin{equation*}
(\b+\bb)_a = 2\left(Y+\Xi \cdot \k-\trt \Xi\right)_a,\quad
(\b-\bb)_a = -2\left(\nabh \Pi-\nabh_N \Xi-2\th\cdot \Xi+p\cdot  \Th-\Pi p \right)_a,
\end{equation*}
 \begin{equation*}
\rhoc = -K -\frac 14 (\trt)^2+\frac 14 (\trth)^2  ,\quad  \rhodc= -\curlh \Xi.
\end{equation*}
As a consequence, in view of Proposition \ref{prop:bounds-local-existence}, we obtain the following bound, in terms of the $L^2$ norm over spheres on $\Si$, for up to $s+1$ angular derivatives: 
\begin{equation*}
\begin{gathered} \sup_{r\in (1,1+\de)} \| \chih,\chibh,\om,\omb \|_{\H^{s+1}(S_r)}  \lesssim a^\frac 12,\quad 
\sup_{r\in (1,1+\de)} \| \rhoc,\rhodc \|_{\H^s(S_r)} \lesssim \de a,\quad \sup_{r\in (1,1+\de)} \| \widecheck\trch, \widecheck\trchb \|_{\H^{s+1}(S_r)}  \lesssim \de a, \\
\quad \sup_{r\in (1,1+\de)} \| \zeta,\eta, \etab ,\slashed{d} (\log\Om) \|_{\H^{s+1}(S_r)}  \lesssim \de a^\frac 12, \quad
\sup_{r\in (1,1+\de)} \| \b, \bb \|_{\H^s(S_r)}  \lesssim a^\frac 12.
\end{gathered}
\end{equation*}
where $\rhoc$ and $\rhodc$ are defined in \eqref{def:rhoc-rhodc}, and 
\beaa
 \trchc:= \trch-2|u|^{-1},\quad \trchbc:= \trchb+2|u|^{-1}.
\eeaa

\subsection{Extension to the triangular region}\lab{sec:tri-region}

In this subsection, we establish the estimates in the spacetime region $\MM_{\triangle}$. The result of this subsection only requires $\de a \lesssim 1$. We define the following norms:
\begin{equation}\lab{eq:def-norms}
\begin{gathered}
\RR_{\leq N} := \sup_{u,\ub}\, \de^{-\frac 12} a^{-\frac 12} \left(\| \nabh^{\leq N} \b \|_{L^2(H_u)} + \|\nabh^{\leq N} (\rhoc,\rhodc)\|_{L^2(\Hb_{\ub}),L^2(H_u)} + \| \nabh^{\leq N} \bb \|_{L^2(\Hb_{\ub})}\right),\\
\OO_{\leq N}:= \sup_{u,\ub}\, a^{-\frac 12} \| \nabh^{\leq N}(\chih, \chibh, \om,\omb) \|_{L^2(S_{u,\ub})} +\de^{-1} a^{-\frac 12} \| \nabh^{\leq N}(\zeta,\etab,\eta) \|_{L^2(S_{u,\ub})}\\
+\de^{-1} a^{-1} \| \nabh^{\leq N}(\trch,\trchb) \|_{L^2(S_{u,\ub})},\\
\OO_{N+1}:=\sup_{u,\ub}\, \de^{-1} a^{-\frac 32} \| \nabh^{ N+1}(\widecheck{\trch}, \widecheck{\trchb}) \|_{L^2(S_{u,\ub})} + \de^{-\frac 12} a^{-\frac 12}
\| \nabh^{N+1} (\chih,\om,\eta,\etab) \|_{L^2(H_u)} \\
+ \de^{-\frac 12} a^{-\frac 12} \|\nabh^{N+1} (\chibh,\omb,\eta,\etab) \|_{L^2(\Hb_{\ub})} .
\end{gathered}
\end{equation}
Here, the $H_u$ and $\underline H_{\ub}$ denote the level hypersurfaces of $u$ and $\ub$ (restricted in $\MM_{\triangle}$), and $S_{u,\ub}$ denotes their intersection spheres.
\begin{proposition}
The spacetime can be extended to $\MM_{\triangle}$, with the bounds 
\begin{equation}\lab{eq:bounds-MM-tri}
\RR_{\leq N}+\OO_{\leq N}+\OO_{N+1}\lesssim 1,\quad  N:=s+1.
\end{equation}
\end{proposition}
\begin{proof}
We make the bootstrap assumptions $\RR_{\leq N}+\OO_{\leq N}+\OO_{N+1}\leq C_b$ with $N:=s+1$, where the $\OO$ and $\RR$ norms are defined in \eqref{eq:def-norms}, and $C_b>0$ is a constant to be determined. We show that we can replace $C_b$ with a constant comparable with $1$, thereby improving the bootstrap assumption by choosing an appropriate $C_b$.

As in An--Luk \cite{AnLuk}, the energy estimates can be derived for the pair $\left(\b,(\rhoc,\rhodc)\right)$ and $\left((\rhoc,\rhodc), \bb\right)$. The main contribution in fact comes from their initial size on $\Si$, whereas the right-hand sides of \eqref{eq:nab3b}-\eqref{eq:nab4bb} only contribute to lower order terms, including those that involve $\OO_{N+1}$---for instance, the contribution from the term $\chih\cdot (\nabh\hot\etab)$, $\chibh\cdot (\nabh\hot\eta)$ in \eqref{eq:nab4P} and \eqref{eq:nab3P} is still lower-order after being integrated:\footnote{We omit the details of the product estimates as they are standard (see e.g., \cite{AnLuk}) and focus on the size analysis.}
\beaa
\iint \chih\cdot (\nabh\hot\etab)\cdot \rhoc \sim \de^2 \cdot C_b a^\frac 12 \cdot C_b a^\frac 12 \cdot C_b a^\frac 12= C_b^3 \de^2 a^\frac 32 \ll (\de^\frac 12 a^\frac 12)^2,
\eeaa
where we plugged in the schematic bounds consistent with $|\rhoc,\chih,\nabh\etab| \lesssim C_b a^\frac 12$. Therefore, we obtain the estimates 
\beaa
\| \nabh^{\leq N}\b \|_{L^2(H_u)} + \|\nabh^{\leq N}(\rhoc,\rhodc)\|_{L^2(\Hb_{\ub})} \lesssim \de^\frac 12 a^\frac 12,\quad \|\nabh^{\leq N}(\rhoc,\rhodc)\|_{L^2(H_u)} + \| \nabh^{\leq N}\bb \|_{L^2(\Hb_{\ub})} \lesssim \de^\frac 12 a^\frac 12,
\eeaa
which yields $\RR_{\leq N}\lesssim 1$. With these estimates established, it is also straightforward to integrate the Bianchi equations to obtain the lower-order bounds on spheres:
\beaa
\| \nabh^{\leq N-1} (\b,\bb) \|_{L^2(S_{u,\ub})}\lesssim a^\frac 12,\quad \| \nabh^{\leq N-1} (\rhoc,\rhodc) \|_{L^2(S_{u,\ub})}\lesssim \de a.
\eeaa

With the curvature quantities improved, we can integrate the null structure equations in Section \ref{sec:null-str-Bianchi-equations} to improve the estimate for all Ricci coefficients. We in fact first improve the estimates for $\chih$ and $\chibh$ by applying the Hodge estimates to \eqref{eq:Codazzi-chih}, \eqref{eq:Codazzi-chibh}:
\begin{equation}\lab{eq:estimate-chih-chibh}
\| \nabh^{\leq N} (\chih,\chibh) \|_{L^2(S_{u,\ub})}\lesssim \| \nabh^{\leq N-1} (\b,\bb)\|_{L^2(S_{u,\ub})}+ a^\frac 12 \lesssim a^\frac 12.
\end{equation}
The estimate for the remaining Ricci coefficients up to $\nabh^{\leq N}$ is straightforward, in view of the fact that they satisfy one of the following schematic equations\footnote{In this proof, we denote by $\Ga$ a schematic Ricci coefficient, and by $R$ a Weyl curvature component.}
\beaa
\nabh_3 \Ga= \Ga\cdot \Ga+R,\quad \nabh_4 \Ga=\Ga\cdot \Ga+R.
\eeaa
The contribution from $R$ is at most $\de a^\frac 12$ by the improved curvature estimates $\RR_{\leq N}\lesssim 1$ we have just obtained, and the contribution from $\Ga\cdot \Ga$ is lower-order for the estimates of $\eta$, $\etab$, $\om$, $\omb$, and exactly the allowed size for $\trch$, $\trchb$ using \eqref{eq:estimate-chih-chibh}. This shows $\OO_{\leq N} \lesssim 1$. 

The $N+1$-th order estimate follows from a similar procedure as in \cite[Section 7]{AnLuk}. We first improve the estimate for $\trch$, $\trchb$ by integrating the Raychaudhuri equations \eqref{eq:e4Raych}, \eqref{eq:e3Raych} by simply using the bootstrap assumption on $\nabh^{\leq N+1}(\chih,\chibh)$; note that we make a weaker assumption on the bounds for $\trch$, $\trchb$ compared with their lower-order controls. We then apply the Hodge estimates again to \eqref{eq:Codazzi-chih}, \eqref{eq:Codazzi-chibh} to improve $\nabh^{\leq N+1}(\chih,\chibh)$. Now, as in \cite{AnLuk}, we introduce the renormalized quantities $\ka$, $\underline\ka$, $\mu_\eta$, and $\underline\mu_{\etab}$ defined as follows:
\begin{equation*}
\begin{gathered}
\ka:= \nabh\om+\dual\nabh \om^\dagger -\frac 12 \b, \quad \nabh_3\om^\dagger=\frac 12\rhodc,\quad \underline{\ka}:= -\nabh\omb+\dual\nabh\omb^\dagger-\frac 12\bb,\quad \nabh_4 \omb^\dagger=\frac 12 \rhodc,\\
 \mu_\eta:=-\divh\eta+\rhoc,\quad \underline{\mu}_{\etab}:=-\divh\etab+\rhoc.
 \end{gathered}
\end{equation*}
These quantities obey transport equations free of differentiated curvature components, see details in \cite{AnLuk}:
\beaa
\nabh_3 \ka, \nabh_4 \mu_\eta, \nabh_3 \underline \mu_{\etab}, \nabh_4\underline \ka = \Ga\cdot \nabh \Ga+\Ga\cdot \Ga\cdot \Ga+\Ga\cdot R+\nabh \Ga, 
\eeaa
and the $\Ga\cdot \Ga\cdot \Ga$ terms always have at least one $\eta$ or $\etab$ factor. Therefore, we have the estimates
\beaa
\| \nabh^{\leq N} (\ka, \underline\mu_{\etab}) \|_{L^2(H_u)}, \| \nabh^{\leq N} (\mu_\eta,\underline\ka) \|_{L^2(\underline H_{\ub})} \lesssim (\de \cdot \de \cdot C_b^2  a)^\frac 12 \ll \de^\frac 12 a^\frac 12.
\eeaa
Then, using Hodge estimates, we recover the $N+1$-th order estimate for $\eta$, $\etab$, $\om$, and $\omb$. 
Therefore, we have $\OO_{N+1}\lesssim 1$, and the proposition follows.
\end{proof}

\subsection{Semi-global existence}\lab{sec:semi-global-existence}

{\bf Characteristic data on $H_{-1}$.}
The bounds in \eqref{eq:bounds-MM-tri} hold, in particular, on $H_{-1}$, the outgoing null cone emanating from $\{r=1\}$ on $\Si$. In fact, we only need to know $\chih$, $\om$, and $\etab$, since other quantities are either non-intrinsic on $H_{-1}$ or they can be determined by the well-known procedure for the characteristic initial value problem, see \cite{Chr1}.

{\bf Semi-global existence.} 
The argument is then reduced to the traditional characteristic problem.
We indicate how this fits into the proof in \cite{CK-TS}. The only
difference is that here $\om\sim a^\frac 12$ on $H_{-1}$, whereas in
the usual setting one imposes $\om=0$ on $H_{-1}$. Examining Proposition
5.5 in \cite{CK-TS}, where $\om$ is estimated through
\beaa
\nabh_3 \om = \rho+|\zeta|^2,
\eeaa
we obtain, instead of
$\om\sim \de^\frac 12 a^\frac 12 |s|^{-\frac 32}$,
the bound
\beaa
\om\sim a^\frac 12+\de^\frac 12 a^\frac 12 |s|^{-\frac 32}.
\eeaa
This additional $a^\frac 12$ term is not amplified as $|s|$ becomes
small. Moreover, in the estimates in which it appears, $\om$ enters
only through the $\nabh_4$ equations. Hence, after integration, the new
contribution is of size $\de a^\frac 12 \ll 1$.

\subsection{Formation of trapped surfaces}\lab{sec:formation-of-trapped-surfaces}

{\bf Lower bound.} 
In view of the bounds \eqref{eq:bounds-MM-tri}, the structure equation 
\beaa
\nabh_3 \chih+\frac 12 \trchb \chih=\nabh\hot \eta+2\omb\chih-\frac 12 \trch\chibh+\eta\hot\eta
\eeaa 
on $\MM_{\triangle}$ implies that the $\chih$ is almost conserved, along the flow of $e_3$, from $\Si$ to $H_{-1}$, modulo $O(\de a)$ errors. Hence, following the anisotropic result of Klainerman--Luk--Rodnianski \cite{KLR} (see also An--Han \cite{AnHan}), the main assumption on $\chih=\thh+\Thh$ on $\Si$ reads
\bea\lab{eq:lower-bound-condition-free-scalars}
\sup_{\vth\in \mathbb{S}^2}\int_1^{1+\de} |\chih(r,\vth)|^2 dr \geq \de a.
\eea
\begin{remark}
Note that we have the relation
\beaa
\d_1 \d_2 \chih= 
-(\BB_{sp}, \BBd_{sp})+(\KK_{sp},-\KKd_{sp})
+O(\de a).
\eeaa
Therefore, it suffices to let the symmetric traceless $2$-tensor $\chih_0$ determined by
\beaa
\d_1\0\d_2\0 \chih_0=-(\BB_{sp}-\KK_{sp},\BBd_{sp}+\KKd_{sp})
\eeaa 
satisfy \eqref{eq:lower-bound-condition-free-scalars}. Here $\d_1\0$ and $\d_2\0$ are the corresponding Hodge operators defined with respect to the round metric $\gz:=r^2 \gacr$.
In particular, we can take $k\equiv 0$ and let the $\chih_0$ determined by $(\BB_{sp},\BBd_{sp})$ verify the requirement, i.e.,  our class of data that leads to trapped surface formation contains time-symmetric initial data.
\end{remark}

\appendix
\section{The case \texorpdfstring{$a\approx \de^{-1}$}{a=1/de}}
\lab{appendix:a-de-inverse}

In Section \ref{sec:construction-Cauchy-data}, we prescribe the free
scalars $(\BB_{sp},\BBd_{sp}, \KK_{sp}, \KKd_{sp})$ and denote their size
by $a^\frac 12>0$. The local existence result, Theorem
\ref{thm:local-existence-precise}, gives the construction on a
short-pulse interval of length $\eps_{loc}$, chosen in terms of the size
of the data. In particular, the argument there closes in the regime $\eps_{loc} a\ll 1$.
Thus $\eps_{loc}$ itself should be viewed as the short-pulse scale in
that construction.

What is not covered directly by that argument is the critical situation
in which one prescribes a short-pulse scale $\de$ and takes
$a\approx \de^{-1}$, i.e., $\de a\approx 1$. 
This is precisely the original short-pulse regime of Christodoulou
\cite{Chr1}. In this appendix, we provide the construction of this type
of short-pulse data in the spacelike setting.

\begin{proposition}\lab{prop:local-existence-Chr}
Fix a small constant $\de>0$. Given $(\BB_{sp},\BBd_{sp}, \KK_{sp}, \KKd_{sp})$ of size $\de^{-\frac 12}$ on $(1,2)\times \mathbb{S}^2$, there exists a unique initial data set solving \eqref{ece} on $(1,1+\de)\times \mathbb{S}^2$ such that all conditions stated in Theorem \ref{thm:local-existence-precise} hold.
\end{proposition}

Before proving Proposition \ref{prop:local-existence-Chr}, we make the following remarks.
\begin{enumerate} 
\item\lab{rem1}
We proceed similarly to \cite{CK26} to construct in the gauge where the metric can be expressed in the form \eqref{eq:metric-abu}, together with the conditions \eqref{eq:ell=1-condition-u-solution1}-\eqref{eq:ell=1-condition-u-solution3}. The main difference from the previous proof is that the expansions $\trth$ and $\trt$ acquire the following contributions respectively, in view of the equations \eqref{eq:thc} and \eqref{eq:trt}:
\beaa
-\frac 12 \int_1^{r} |\thh|^2+|\Thh|^2\, dr' ,\quad -\int_1^{r} \thh\cdot \Thh\, dr'.
\eeaa
Both are of size $\de \cdot (\de^{-\frac 12})^2=1$, and hence are not small.
\item
 Recall that in the proof of Theorem \ref{thm:local-existence-precise} given in \cite{CK26}, we distinguished between small quantities $\Ga$ and large quantities $\widetilde \Ga\setminus \Ga$:
 \beaa
 \Ga= \{u, \thc, \trt, \slashed{d}(\log\ah), \Xi\} ,\quad \widetilde \Ga\setminus \Ga=\{ b, \thh, \Thh, \Pi, Y\}.
 \eeaa
According to Remark \ref{rem1}, we now further divide them into
 \beaa
 \Ga_{small}= \{u, \slashed{d}(\log\ah), \Xi\} , \quad \Ga_{medium}=\{\trth,\trt\},\quad \Ga_{large}=\{ b, \thh, \Thh, \Pi, Y\}.
 \eeaa
\item We note that, apart from the transport equations of $\thc$ and $\trt$, the product terms in all remaining equations have the form either $\Ga_{small}\cdot\Ga_{large}$ or $\Ga_{medium}\cdot \Ga_{medium}$. Moreover, as in the original proof in \cite{CK26}, the estimate for $\Ga_{small}$ is independent of $\Ga_{large}$; see, for instance, \eqref{eq:laph-u-estimate-part}-\eqref{eq:u-ell=1-estimate-part} and \eqref{eq:nu-estimate}-\eqref{eq:curl-Xi-estimate}. To also avoid the dependence on $\Ga_{medium}$, we modify the definition of the gauge scalar $\mu$:
\beaa
\mu':=-\laph(\log\ah)+\Kc,
\eeaa
and, impose the condition $\mu'_{\ell\geq 1}=0$ instead of $\mu_{\ell\geq 1}=0$. Then, in view of \eqref{eq:laph-ao-estimate-part}-\eqref{eq:ao-ell=0-estimate-part}, the estimate for $\ao$ is also independent of $\Ga_{medium}$.
\item
The most important point is that the perturbed Gauss curvature $\Kc$ remains small, of size $O(\de^{\frac 12})$. This ensures that the gauge choice based on effective uniformization is still valid, and that the relevant Hodge operators remain close to those associated with the round metric.
\end{enumerate}

\begin{proof}[Proof of Proposition \ref{prop:local-existence-Chr}]
We prove Proposition \ref{prop:local-existence-Chr} by an iteration argument. We omit the details of the iteration scheme and the limit identification, since they are essentially the same as in \cite{CK26}. Instead, we focus on the improvement of the estimates needed to close the iteration. We assume
\beaa
\| \nabh^{\leq 2} \ao,\nabh^{\leq 1}\Xi,\Kc \|_{\H^{s}(S_r)} \lesssim C_b \de^\frac 12,\quad \| \thc, \trt \|_{\H^{s+1}(S_r)} \lesssim C_b,\quad \| \thh, \Thh, \Pi \|_{\H^{s+1}(S_r)} \lesssim C_b\de^{-\frac 12},
\eeaa
where $C_b>0$ is the boundedness constant which we seek to improve.
Schematically, this reads 
\beaa
\| \Ga_{small} \|+ \de^\frac 12 \| \Ga_{medium}\| + \de \| \Ga_{large} \|\lesssim C_b\de^\frac 12.
\eeaa

We recall the HCS equations in \cite{CK26}, but with the more careful categorization of the $\Ga$ terms (also note the change from $\mu$ to $\mu'$):\footnote{We use a two-level schematic notation. The labels
``small'', ``medium'', and ``large'' refer to the size classes introduced
above, while the subscript on $\Ga_i$ was introduced in \cite{CK25,CK26} that indicates the differentiability order.  More precisely, we define
\begin{equation*}
\Ga_0=\{\ao,u,b\},\quad \Ga_1=\{r^{-1} \Ga_0,\thc,\, \slashed{d} (\log\ah),\Xi,\trt,\thh,\Thh,\Pi \},\quad \Ga_2=\{r^{-1}\Ga_1,\nabh \Ga_1,\Kc,Y\}.
\end{equation*}
The combined notations are understood as the intersection of the two classes.}
\begin{align}
\lab{eq:N-thc-estimate-part}
N \thc &= \mu'_{\ell=0} -3r^{-1}\thc-2 r^{-2}\ao -\frac 12 |\thh|^2-\frac 12 |\Thh|^2 \\
&\quad \notag+(\Ga_1)_{medium}\cdot (\Ga_1)_{large}, \\
\lab{eq:N-K-ell-neq-1-estimate-part}
N \Kc &= r^{-1} \mu'_{\ell=0}- (\BB_{sp}+\Bb_{\ell\leq 1}) -3r^{-1}\Kc -2 r^{-3} \ao -r^{-2} \thc \\
&\notag\quad +2{\slashed{d}(\log\ah)\cdot Y}-\thh\cdot (\nabh\hot \slashed{d}(\log\ah))\\
&\notag \quad +(\Ga_1)_{medium}\cdot (\Ga_2)_{medium},
\end{align}
\begin{align}
\lab{eq:laph-ao-estimate-part}
(\laph (\log\ah))_{\ell\geq 1}&= \Kc_{\ell\geq 1},\\
\lab{eq:ao-ell=0-estimate-part}
\overline{\ao} &=0,\\
\lab{eq:laph-u-estimate-part}
-\laph_{\gacr} u-2u &=r^2 \Kc+r^2\Kc \cdot u,\\
\lab{eq:u-ell=1-estimate-part}
u_{\ell=1} &=u\cdot u.\\
\lab{eq:d1-d2-thh-estimate}
\d_1 \d_2 \thh &= \big(\frac 12 \laph\thc,0\big)-(\BB_{sp}+\Bb_{\ell\leq 1},\BBd_{sp}+\Bbd_{\ell\leq 1}), \\
\d_1 Y &=  (\BB_{sp}, \BBd_{sp})+(\Bb_{\ell\leq 1},\Bbd_{\ell\leq 1}) ,
\end{align}
\begin{align}
\lab{eq:N-trt-estimate}
N\trt &= 2 r^{-1} \Pi  - r^{-1} \trt -\thh \cdot \Thh +(\Ga_1)_{medium}\cdot (\Ga_1)_{large} ,\\
\lab{eq:d1-d2-Thh-estimate}
\d_1 \d_2 (\ah \Thh)
 &= (\KK_{sp}, -\KKd_{sp})+(\Kk_{\ell\leq 1},\Kkd_{\ell\leq 1})+ (\Ga_1)_{large}\cdot (\Ga_2)_{small}\\
 &\notag \quad +(\Ga_1)_{small}\cdot (\Ga_2)_{large}+(\Ga_1)_{medium} \cdot (\Ga_2)_{medium}, \\
\lab{eq:nu-estimate}
\divh \Xi &=0, \\
\lab{eq:curl-Xi-estimate}
\ah N \curlh \Xi  &=-4r^{-1} \curlh \Xi+ \KKd_{sp}-\Kkd_{\ell\leq 1},\\
\lab{eq:estimate-Pi}
\laph \left(\Pi +\frac 12 \trt \right) &= \KK_{sp} + \Kk_{\ell\leq 1},\\
\lab{eq:average-Pi-estimate}
\overline{\Pi}&= -\frac 12 \overline{\trt}, 
\end{align}
\begin{align}
\lab{eq:nab-hot-b-estimate}
\nabz\hot b^{\flat,0} &=-2 e^{-2u} \ah\, \thh,\\
\lab{eq:div-b-ell=1}
(\divz (e^{2u} b^{\flat,0}))_{\ell=1} & = -\left(e^{2u} (\ah\, \trth -2r^{-1})\right)_{\ell=1},\\
\lab{eq:curl-b-ell=1}
(\curlz (e^{2u} b^{\flat,0}))_{\ell=1}&=0.
\end{align}

Here, $(\Bb_{\ell\leq 1}, \Bbd_{\ell\leq 1})$ and $(\Kk_{\ell\leq 1},\Kkd_{\ell\leq 1})$ are the $\ell\leq 1$ corrections naturally arise for the solvability of the Codazzi equations, see details in \cite{CK25,CK26}.

We proceed as follows:\footnote{Throughout the proof, the implicit constant in $\lesssim$ is independent of $C_b$.}
\begin{itemize}
\item We first integrate the transport equation for $\Kc$,
\eqref{eq:N-K-ell-neq-1-estimate-part}. Projecting
\eqref{eq:d1-d2-thh-estimate} to $\ell\leq 1$ and using the rough bound
$\| \thc \|_{\H^{s+1}(S_r)} \lesssim C_b$, we deduce
\begin{equation*}
| \Bb_{\ell\leq 1},\Bbd_{\ell\leq 1} | \lesssim C_b.
\end{equation*}
Together with the size $\BB_{sp}\sim \de^{-\frac 12}$ of the free scalars, the Hodge equation for $Y$ then gives
\begin{equation*}
\|Y\|_{\H^s(S_r)} \lesssim \de^{-\frac 12}+C_b \lesssim \de^{-\frac 12}.
\end{equation*}
The remaining nonlinear terms in \eqref{eq:N-K-ell-neq-1-estimate-part}
are of the form $\Ga_{small}\cdot \Ga_{large}$ or
$\Ga_{medium}\cdot \Ga_{medium}$, and are bounded by $C_b^2$.
Therefore, integrating \eqref{eq:N-K-ell-neq-1-estimate-part}, we obtain
\begin{equation*}
\| \Kc \|_{\H^s(S_r)} \lesssim \de\cdot \de^{-\frac 12}+\de C_b+\de C_b^2 \lesssim \de^\frac 12.
\end{equation*}
The estimate for $u$ then follows from
\eqref{eq:laph-u-estimate-part}-\eqref{eq:u-ell=1-estimate-part} and the improved bound for $\Kc$:
\begin{equation*}
\| u \|_{\H^{s+2}(S_r)} \lesssim \de^\frac 12.
\end{equation*}
\item Using the equation \eqref{eq:laph-ao-estimate-part}, which is the relation implied by the condition $\mu'_{\ell\geq 1}=0$, together with \eqref{eq:ao-ell=0-estimate-part}, we obtain the improved estimate for $\ao$:
\begin{equation*}
\| \ao \|_{\H^{s+2}(S_r)} \lesssim \|\Kc\|_{\H^s(S_r)} \lesssim \de^\frac 12.
\end{equation*}
\item The Codazzi equation \eqref{eq:d1-d2-thh-estimate} gives a direct
estimate for $\thh$. The principal source is the prescribed free scalar
$(\BB_{sp},\BBd_{sp})$, of size $\de^{-\frac 12}$. The $\ell\leq 1$ correction is
bounded by $C_b$, and the $\laph\trth$ term is also bounded by $C_b$.
Therefore,
\begin{equation}\lab{eq:improved-estimate-thh}
\| \thh \|_{\H^{s+1}(S_r)} \lesssim \de^{-\frac 12}+C_b \lesssim \de^{-\frac 12}.
\end{equation}
This also yields the estimate for $b$, in view of \eqref{eq:nab-hot-b-estimate}, \eqref{eq:div-b-ell=1}, and \eqref{eq:curl-b-ell=1}:
\begin{equation*}
\|b\|_{\H^{s+2}(S_r)} \lesssim \de^{-\frac 12}+C_b^2 \lesssim \de^{-\frac 12}.
\end{equation*}
\item The estimate for $\Thh$ using
\eqref{eq:d1-d2-Thh-estimate} is similar. The principal source is
$(\KK_{sp},-\KKd_{sp})$, of size $\de^{-\frac 12}$, while the remaining product
terms are either $\Ga_{small}\cdot\Ga_{large}$ or
$\Ga_{medium}\cdot\Ga_{medium}$. Hence
\begin{equation}\lab{eq:improved-estimate-Thh}
\| \Thh \|_{\H^{s+1}(S_r)} \lesssim \de^{-\frac 12}+C_b^2 \lesssim \de^{-\frac 12}.
\end{equation}
\item To estimate $\thc$,  we use \eqref{eq:improved-estimate-thh} and \eqref{eq:improved-estimate-Thh} 
and integrate \eqref{eq:N-thc-estimate-part} to deduce
\begin{equation*}
\| \thc \|_{\H^{s+1}(S_r)} \lesssim \de\cdot \de^{-\frac 12}\cdot \de^{-\frac 12} =1.
\end{equation*}
\item It is straightforward to obtain the estimate for $\Xi$
\begin{equation*}
\| \Xi \|_{\H^{s+1}(S_r)} \lesssim \de \cdot (\de^{-\frac 12}+ C_b^2) \lesssim \de^\frac 12,
\end{equation*}
and the estimate for $\Pi+\frac 12 \trt$
\begin{equation*}
\| \Pi+\frac 12 \trt \|_{\H^{s+1}(S_r)}\lesssim \de^{-\frac 12}.
\end{equation*}
The estimate for $\trt$ is similar to that for $\thc$. We obtain
\begin{equation*}
\| \trt \|_{\H^{s+1}(S_r)} \lesssim 1.
\end{equation*}
\end{itemize}
This closes the estimates. The existence result then follows by the same iteration procedure as in \cite{CK26}.
\end{proof}

{\bf The extension to asymptotically flat data.}
Having constructed Christodoulou-type data in the short-pulse region, we now discuss the question of extending such data to spacelike infinity. The main difference again comes from the fact that $\thc$ and $\trt$, even after subtracting their Minkowski values, are no longer small. Indeed, we have shown that
\beaa
\trth=\frac{2}r-\frac 12 \int_1^r \left(|\thh_0|^2 +|\Thh_0|^2\right)\,dr'+O(\de^\frac 12),
\quad
\trt=-\int_1^r \thh_0\cdot \Thh_0\,dr'+O(\de^\frac 12).
\eeaa
This creates a major obstruction for the global problem.\footnote{In our language, to extend the solution to $r=\infty$ through a bootstrap argument, one needs smallness relative to a suitable reference solution.} This is the essential reason why Li--Yu \cite{LiYu} impose the symmetry condition, which in the characteristic setting reads
\beaa
\int_0^\de |\chih|^2(\ub,\vth)\,d\ub
\quad
\text{is independent of $\vth\in\mathbb S^2$.}
\eeaa
Under this condition, the outgoing expansion $\trch$ becomes close to the corresponding Schwarzschild value. In our setting, to obtain such closeness, one needs to impose analogous conditions, for instance
\beaa
\int_1^{1+\de} \thh_0\cdot \Thh_0 \, dr=O(\de^\frac 12),
\quad
\int_1^{1+\de} \left(|\thh_0|^2 +|\Thh_0|^2\right)(r,\vth)\,dr
\quad
\text{is independent of $\vth\in\mathbb S^2$.}
\eeaa
Apart from this additional requirement, the proof follows the same strategy, and we do not pursue the details here.


\begin{thebibliography}{99}
\bibitem{An12}
X.~An,
\newblock \emph{Formation of trapped surfaces from past null infinity},
\newblock arXiv:1207.5271.

\bibitem{An19}
X.~An,
\newblock \emph{A scale-critical trapped surface formation criterion: a new proof via signature for decay rates},
\newblock \emph{Ann. PDE} \textbf{8} (2022), no.~1, Paper No.~3.

\bibitem{AnHan}
X.~An and Q.~Han,
\newblock \emph{Anisotropic dynamical horizons arising in gravitational collapse},
\newblock arXiv:2010.12524.

\bibitem{AnLuk}
X.~An and J.~Luk,
\newblock \emph{Trapped surfaces in vacuum arising dynamically from mild incoming radiation},
\newblock \emph{Adv. Theor. Math. Phys.} \textbf{21} (2017), no.~1, 1–120.

\bibitem{BieriChrusciel2017}
L.~Bieri and P.~T.~Chru\'sciel,
\newblock \emph{Future-complete null hypersurfaces, interior gluings, and the Trautman--Bondi mass},
\newblock \emph{Proc. Harv. CMSA} \textbf{1} (2017), 1--31.



\bibitem{Gauge-scalars-preparation}
X.~Chen,
\newblock \emph{Gauge scalars for the Einstein constraint equations},
\newblock in preparation.


\bibitem{CK25}
X.~Chen and S.~Klainerman,
\emph{Solving the constraint equation for general free data},
arXiv:2512.22704.

\bibitem{CK-TS}
X.~Chen and S.~Klainerman,
\emph{Formation of Trapped Surfaces in Geodesic Foliation}, 
\newblock \emph{Comm. Math. Phys.}
\textbf{407}, 98 (2026).

\bibitem{CK26}
X.~Chen and S.~Klainerman,
\emph{Forward Construction of Vacuum Initial Data with Borderline Decay},
arXiv:2606.08716.


 \bibitem{Chr1} D. Christodoulou, \textit{The formation of Black Holes  in General Relativity}, EMS Monographs in Mathematics, 2009.
 
 
 
 \bibitem{C-K}D. Christodoulou and S. Klainerman, \textit{The nonlinear stability of the Minkowski space},
 Princeton University Press, 1993.


\bibitem{ChruscielDelay2003}
P.~T.~Chru\'sciel and E.~Delay,
\newblock \emph{On mapping properties of the general relativistic constraints operator in weighted function spaces, with applications},
\newblock \emph{M\'em. Soc. Math. Fr. (N.S.)} \textbf{94} (2003), vi+103~pp.



\bibitem{Corvino2000}
J.~Corvino,
\newblock \emph{Scalar curvature deformation and a gluing construction for the Einstein constraint equations},
\newblock \emph{Comm. Math. Phys.} \textbf{214} (2000), no.~1, 137--189.

\bibitem{CorvinoSchoen}
J.~Corvino and R.~M.~Schoen,
\newblock \emph{On the asymptotics for the vacuum Einstein constraint equations},
\newblock \emph{J. Differential Geom.} \textbf{73} (2006), no.~2, 185--217.



\bibitem{GiorgiShenWanBosonStars}
E.~Giorgi, D.~Shen, and J.~Wan,
\newblock \emph{Cauchy data for multiple collapsing boson stars},
\newblock arXiv:2512.01844.

 \bibitem{KLR} S. Klainerman, J.  Luk and  I. Rodnianski,   \textit{ A   Fully Anisotropic Mechanism for Formation of Trapped Surfaces in Vacuum},  
Inventiones,  Volume 198, Issue 1 (2014), Page 1-26.


\bibitem{KR05}
S.~Klainerman and I.~Rodnianski,
\emph{Causal geometry of Einstein-vacuum spacetimes with finite curvature flux},
Invent. Math. \textbf{159} (2005), 437--529.


\bibitem{KRodn}
S.~Klainerman and I.~Rodnianski,
\emph{On the formation of trapped surfaces},
Acta Math. \textbf{208} (2012), 211--333.


\bibitem{LiMei-CollapsingVacuum}
J.~Li and H.~Mei,
\emph{A construction of collapsing spacetimes in vacuum},
Comm. Math. Phys. \textbf{378} (2020), no.~2, 1343--1389.

\bibitem{LiYu}
J.~Li and P.~Yu,
\newblock \emph{Construction of Cauchy data of vacuum Einstein field equations evolving to black holes},
\newblock \emph{Ann. Math.} (2) \textbf{181} (2015), 699--768.


\bibitem{LR15}
J.~Luk and I.~Rodnianski,
\emph{Local propagation of impulsive gravitational waves},
Comm. Pure Appl. Math. \textbf{68} (2015), no. 4, 511--624.


\bibitem{LR17}
J. Luk and I. Rodnianski,
Nonlinear interaction of impulsive gravitational waves for the vacuum Einstein equations,
{\em Cambridge J. Math.} {\bf 5} (2017), no. 4, 435--570.


\bibitem{MaoOhTao2023}
Y.~Mao, S.-J.~Oh, and Z.~Tao,
\newblock \emph{Initial data gluing in the asymptotically flat regime via solution operators with prescribed support properties},
\newblock arXiv:2308.13031.



\bibitem{Shen23}
D.~Shen,
\newblock \emph{Stability of the Minkowski space with minimal decay},
\newblock arXiv:2310.07483.


\bibitem{Shen24}
D.~Shen,
\newblock Exterior stability of Minkowski spacetime with borderline decay,
\newblock {\em Ann. Sci. \'Ec. Norm. Sup\'er.}, {\bf 59} (2026), no.~4, 867--907.

\bibitem{ShenWan2025}
D.~Shen and J.~Wan,
\newblock \emph{Formation of multiple black holes from Cauchy data},
\newblock arXiv:2506.16117.

\bibitem{ShenWan26ADM}
D.~Shen and J.~Wan,
\newblock Cauchy data for formation of multiple black holes with prescribed ADM parameters,
\newblock arXiv:2601.01517.




\end{thebibliography}
\end{document}